\documentclass[a4paper,USenglish]{lipics-v2016-arxiv}
\usepackage{amssymb,amsmath,amsthm}
\usepackage{hyperref}

\usepackage{tikz}
\usetikzlibrary{arrows,shapes,plotmarks}

\usepackage{xspace}
\usepackage{mathtools}
\usepackage{enumitem}
\setlist{noitemsep}

\newcommand{\cl}[1]{\ensuremath{\mathsf{#1}}}
\renewcommand{\P}{\cl{P}\xspace}
\newcommand{\NP}{\cl{NP}\xspace}
\newcommand{\FPT}{\cl{FPT}\xspace}
\newcommand{\Wone}{\cl{W[1]}\xspace}
\newcommand{\Wtwo}{\cl{W[2]}\xspace}

\DeclarePairedDelimiter{\civ}{[}{]}

\DeclarePairedDelimiter{\roiv}{[}{)}

\DeclareMathOperator{\reali}{Reali}

\newcommand{\Reals}{\mathbb{R}}
\newcommand{\Ints}{\mathbb{Z}}

\newcommand{\mypara}[1]{\vspace{10pt} \noindent \textbf{\sffamily #1}}
 
\newcommand{\observe}[1]
{\vspace{6pt} \noindent $\triangleright$ {\sffamily Observation.}{\it #1} \vspace{6pt}}

\newcommand{\cK}{\mathcal{K}}
\newcommand{\cS}{\mathcal{S}}
\newcommand{\cI}{\mathcal{I}}
\newcommand{\cG}{\mathcal{G}}
\newcommand{\cC}{\mathcal{C}}
\newcommand{\cD}{\mathcal{D}}
\newcommand{\cF}{\mathcal{F}}
\newcommand{\bu}{\mathbf{u}}
\newcommand{\bb}{\mathbf{b}}
\newcommand{\bp}{\mathbf{p}}
\newcommand{\bR}{\mathbf{R}}

\newcommand{\eps}{\varepsilon}

\renewcommand{\leq}{\leqslant}
\renewcommand{\geq}{\geqslant}

\newcommand{\true}{\textsc{true}}

\newcommand{\graph}{G}

\newcommand{\skb}[1]{{ #1}}

\title{The Dominating Set Problem in Geometric Intersection Graphs\footnote{This research was
supported by the Netherlands Organization for Scientific Research (NWO) under
project no.~024.002.003.}}

\author[1]{Mark de Berg}
\author[2]{S\'andor Kisfaludi-Bak}
\author[3]{Gerhard Woeginger}
\affil[1]{Department of Mathematics and Computer Science, TU Eindhoven, Eindhoven, the Netherlands}
\affil[2]{Department of Mathematics and Computer Science, TU Eindhoven, Eindhoven, the Netherlands}
\affil[3]{Department of Computer Science, RWTH Aachen University, Aachen, Germany}
\authorrunning{M. de Berg, S. Kisfaludi-Bak, and G. Woeginger}
\Copyright{Mark de Berg, S\'andor Kisfaludi-Bak, and Gerhard Woeginger}
\keywords{Broadcast, dominating set, unit disk graphs, range assignment}
 \EventShortTitle{ArXiv version}
 \EventAcronym{ArXiv}
\begin{document}

\maketitle

%%%%%%%%%%%%%%%%%%%%%%%%%%%%%%%%%%%%%%%%%%%%%%%
%%%%%%%%%%%%%%%%%%%%%%%%%%%%%%%%%%%%%%%%%%%%%%%

\begin{abstract}
We study the parameterized complexity of dominating sets in geometric
intersection graphs.
\begin{itemize}
\item In one dimension, we investigate intersection graphs induced by
translates of a fixed pattern Q that consists of a finite number of intervals
and a finite number of isolated points. We prove that Dominating Set on such
intersection graphs is polynomially solvable whenever Q contains at least one
interval, and whenever Q contains no intervals and for any two point pairs in
Q the distance ratio is rational. The remaining case where Q contains no
intervals but does contain an irrational distance ratio is shown to be NP-
complete and contained in FPT (when parameterized by the solution size).
\item In two and higher dimensions, we prove that Dominating Set is contained
in W[1] for intersection graphs of semi-algebraic sets with constant
description complexity. This generalizes known results from the literature.
Finally, we establish W[1]-hardness for a large class of intersection graphs.
\end{itemize}
\end{abstract}

%%%%%%%%%%%%%%%%%%%%%%%%%%%%%%%%%%%%%%%%%%%%%%%
%%%%%%%%%%%%%%%%%%%%%%%%%%%%%%%%%%%%%%%%%%%%%%%
\section{Introduction}
%%%%%%%%%%%%%%%%%%%%%%%%%%%%%%%%%%%%%%%%%%%%%%%

A \emph{dominating set} in a graph $\graph=(V,E)$ is a subset $D\subseteq V$
of vertices  such that every node in $V$ is either contained in $D$ or has
some neighbor in~$D$.  The decision version of the dominating set problem asks
for a given graph $\graph$ and a given integer~$k$, whether $\graph$ admits a
dominating set of size at most~$k$.  Dominating set is a popular and classic
problem in algorithmic graph theory.  It has been studied extensively for
various graph classes; we only mention that it is  polynomially solvable on
interval graphs, strongly chordal graphs, permutation graphs and
co-comparability graphs and that it is \NP-complete on bipartite graphs,
comparability graphs, and split graphs. We refer the reader to the book
\cite{HaHeSl1998} by Hales, Hedetniemi and Slater for lots of comprehensive 
information on dominating sets.

Dominating set is also a model problem in parameterized complexity, as it is
one of the few natural problems known to be \Wtwo-complete (with the solution
size $k$ as natural  parameterization); see~\cite{Downey95}.
In the parameterized setting, dominating set on a concrete graph class
typically is either in \P, \FPT, \Wone-complete, or \Wtwo-complete.
(Note that the problem cannot be  on higher levels of the W-hierarchy, as it
is \Wtwo-complete on general graphs.)

In this paper we study the dominating set problem on geometric intersection
graphs:  Every vertex in $V$ corresponds to a geometric object in $\Reals^d$,
and there is an edge between two vertices if and only if the corresponding
objects intersect. Well-known graph classes that fit into this model are
interval graphs and unit disk graphs.  In $\Reals^1$, Chang~\cite{Chang98} has given a polynomial time algorithm for dominating set in interval graphs and Fellows, Hermelin, Rosamond and Vialette~\cite{Fellows09} have proven
\Wone-completeness for $2$-interval graphs (where the geometric
objects are pairs of intervals). In $\Reals^2$, Marx~\cite{Marx06} has shown that  dominating set is \Wone-hard
for unit disk graphs as well as for unit square graphs. For unit
square graphs the problem is furthermore known to be contained in
\Wone~\cite{Marx06}, whereas  for unit disk graphs this was previously not
known.

\mypara{Our contribution.}{
We investigate the dominating set problem on intersection graphs of 1- and 2-dimensional objects, 
thereby shedding more light on the borderlines between \P and \FPT and \Wone and \Wtwo.}

For 1-dimensional intersection graphs, we consider the following setting.
There is a fixed \emph{pattern}~$Q$, which consists of a finite number of
points and a finite number of closed intervals (specified by their endpoints).
The objects corresponding to the vertices in the intersection graph simply are
\skb{a finite number of}  translates of  this fixed pattern~$Q$.  More
formally, for a real number $x$ we define $Q(x):= x+Q$ to be the pattern $Q$
translated by $x$\skb{, and for the input $\{x_1,\dots, x_n\}$, we consider
the intersection graph defined by the objects $\{x_1+Q,\dots, x_n+Q\}$}. The
class of unit interval graphs arises by choosing $Q=[0,1]$. Our model of
computation is the word RAM model, where real numbers are restricted to a
field $K$  which is a finite extension of the rationals.

\begin{remark}[Machine representation of numbers]
As finite extensions of $\mathbb{Q}$ are finite dimensional vector spaces over
$\mathbb{Q}$,  there exists a basis $b_1,\ldots,b_k$ with $k=[K:\mathbb{Q}]$,
so that any real $x\in K$ is  representable in the form
$x=q_1b_1+q_2b_2+\dots+q_kb_k$ for some $q_1,\ldots,q_k\in\mathbb{Q}$.  As $k$
is fixed, any arithmetic operation that takes $O(1)$ steps on the rationals
will also take $O(1)$ steps on elements of $K$.
\end{remark}

We define the \emph{distance ratio} of two point pairs $(x_1,x_2), (x_3,x_4)
\in \Reals\times \Reals$ as $\frac{|x_1-x_2|}{|x_3-x_4|}$.
We derive the following complexity classification for \textsc{$Q$-Intersection
 Dominating Set}.

\begin{theorem}\label{thm:onedim}
\textsc{$Q$-Intersection Dominating Set} has the following complexity:
\begin{enumerate}[label=(\roman*)]
\item It is in \P if the pattern $Q$ contains at least one interval.
\item It is in \P if the pattern $Q$ does not contain any intervals, and if for
 any two point pairs in $Q$ the distance ratio is rational.
\item It is \NP-complete and in \FPT if pattern $Q$ is a finite point set 
which has at least one irrational distance ratio.
\end{enumerate}
\end{theorem}

In Lemma~\ref{lem:large_pattern_proof} we show that any graph can be
obtained as a  1-dimensional pattern intersection graph for a suitable choice
of pattern $Q$. Consequently \textsc{$Q$-Intersection Dominating Set} is
\Wtwo-complete if the pattern $Q$  is part of the input.

\medskip 
For 2-dimensional intersection graphs, our results are inspired by a
question  that was not resolved in \cite{Marx06}: \emph{``Is dominating set on
unit disk graphs  contained in \Wone?''} We answer this question affirmatively
(and thereby fully settle  the complexity status of this problem).  Our result
is in fact far more general: We show that dominating set is contained in \Wone
whenever the geometric objects in the intersection graph come from a family of
semi-algebraic  sets that can be described by a constant number of parameters.
We also show that this restriction to shapes of constant-complexity is
crucial, as dominating  set is \Wtwo-hard on intersection graphs of convex
polygons with a polynomial number of  vertices.  On the negative side, we
generalize the \Wone-hardness result of Marx~\cite{Marx06} by showing that for
any non-trivial simple polygonal pattern $Q$, the corresponding version  of
dominating set is \Wone-hard.

%-------------------------------------------------------------------------

\section{1-dimensional patterns}

In this section, we study the \textsc{$Q$-Intersection Dominating Set}
problem in $\Reals^1$.  If $Q$ contains an unbounded interval, then all
translates are intersecting; the intersection graph is a clique and the
minimum dominating set is a single vertex. In what follows, we assume that
all intervals in $Q$ are bounded. We define the \emph{span} of $Q$ to be the
distance between its leftmost and rightmost point. We prove
Theorem~\ref{thm:onedim} by studying each claim separately.

\begin{lemma}\label{lem:pattern_with_interval}
\textsc{$Q$-Intersection Dominating Set} can be solved in $O(n^{6w+4})$ time
if $Q$ contains at least one interval, where $w$ is the ratio of the span of
 $Q$ and the length of the longest interval in $Q$.
\end{lemma}

Note that since $Q$ is a fixed pattern, the value of $w$ does not depend on
the input size and so Lemma~\ref{lem:pattern_with_interval} implies Theorem
~\ref{thm:onedim}(i). We translate $Q$ so that its leftmost endpoint
lies at the origin, and we rescale $Q$ so that its longest interval has
length $1$. Consider an intersection graph $\cG$ of a set of translates of
$Q$. The vertices of $\cG$ are $Q(x_i)$ for the given values $x_i$. We call
$x_i$ the \emph{left endpoint} of $Q_i$. Let $+$ also denote the Minkowski
sum of sets: $A+B=\{a+b\;|\;a\in A, b\in B\}$. If $A$ or $B$ is a singleton,
then we omit the braces, i.e., we let $a+B$ denote $\{a\}+B$. In order to prove Lemma~\ref{lem:pattern_with_interval}, we need the following lemma first.

\begin{lemma}\label{lem:sparse_dom_set}
Let $D\subseteq V(\cG)$ be a minimum dominating set and let $X(D)$ be the set
of left endpoints corresponding to the patterns in $D$. Then for all $y\in
\Reals$ it holds that $|X(D)\cap\civ{y,y+w}|\leq 3w$.
\end{lemma}

\begin{proof}
We prove this lemma first for unit interval graphs (where $Q$ consists of a
single interval). The following observation is easy to prove.

\observe{
In any unit interval graph there is a minimum dominating set whose intervals
do not overlap.}{}
% {Take a minimum dominating set $D$, and suppose it has two
% overlapping intervals $I_1$ and $I_2$, such that the left endpoint of $I_1$
% lies to the left of the left endpoint of $I_2$. The set $D \setminus \{I_2\}$
% does not dominate every interval. Let $H$ be the set of undominated intervals
% (the set of intervals that have no neighbor in $D \setminus \{I_2\}$). The
% intervals in $H$ lie to the right of $I_1$ (since they were previously
% intersected by $I_2$, but they are disjoint from $I_1$). The rightmost
% interval $I_H \in H$ intersects all intervals of $H$, since all intervals of
% $H$ have their left endpoints in $I_2 \setminus I_1$, an interval of length
% less than $1$. Thus, $(D \setminus \{I_2\})\cup \{I_H\}$ is a minimum
% dominating set. Repeating this operation on overlapping intervals terminates
% because the sum of the left endpoints of the dominating set strictly
% increases after each step. Therefore, this results in an overlap-free minimum
% dominating set.}

Notice that the lemma immediately follows from this claim in case of unit
interval graphs since then $|X(D)\cap\civ{y,y+1}|\leq 1 < 3=3w$. Let $Q$ be
any other pattern, and suppose that $|X(D)\cap\civ{y,y+w}| \geq 3w+1$. The
patterns starting in $\civ{y,y+w}$ can only dominate patterns with a left
endpoint in $\civ{y-w,y+2w}$, a window of width $3w$. Let $H$ be the set of
patterns starting in $\civ{y-w,y+2w}$ (see Figure~\ref{fig:windowpattern}. Let
$I$ be a unit interval of $Q$, and let $U$ the set of unit intervals that are
the translates of $I$ in the patterns of $H$. Notice that $X(U)$ is a point
set that is also in a window of length $3w$. By the claim above, we know that
the interval graph $\cG(U)$ defined by $U$ has a dominating set that contains
non-overlapping intervals, in particular, a dominating set $D_U$ of size at
most $3w$. Since $\cG(U)$ corresponds to a spanning subgraph of $\cG(H)$, the
patterns $D_U^H$ corresponding to $D_U$ in $H$ form a dominating set of
$\cG(H)$. Thus, $(D \setminus H) \cup D_U^H$ is a dominating set of our
original graph that is smaller than $D$, which contradicts the minimality of
$D$.
\end{proof}

\begin{figure}
\begin{center}
\includegraphics[scale=0.75]{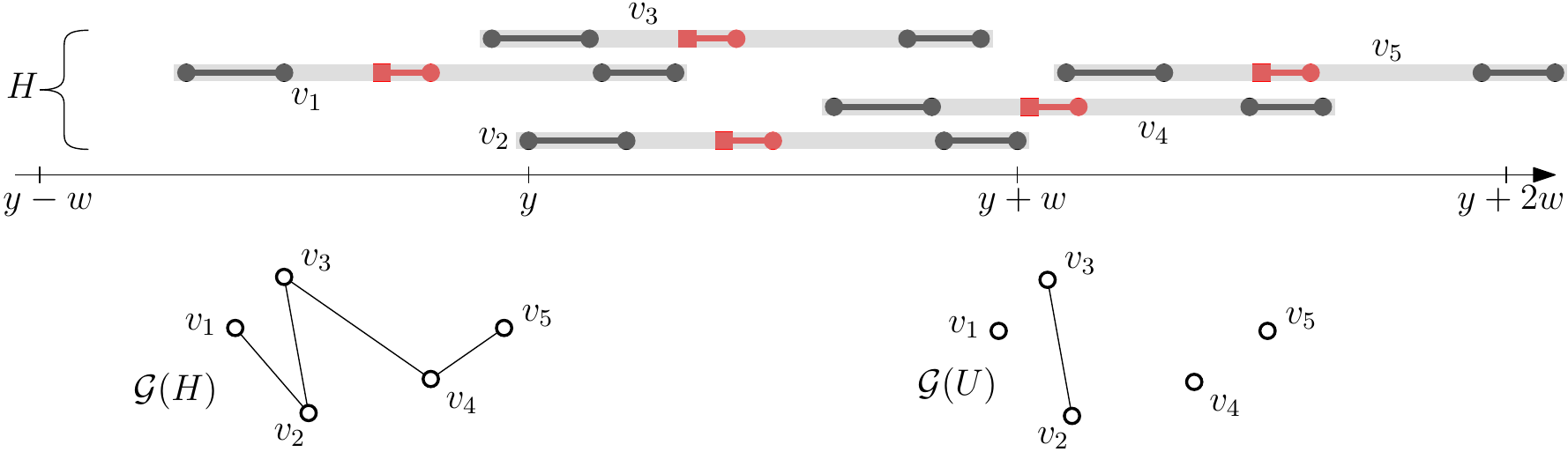}
\end{center}
\caption{Patterns in a window $\civ{y-w,y+2w}$. Intervals of $U$ are red.}
\label{fig:windowpattern}
\end{figure}

We can now move on to the proof of Lemma~\ref{lem:pattern_with_interval}.
\begin{proof}
We give a dynamic programming algorithm. We translate our input so that the
left endpoint of the leftmost pattern is $0$. Moreover, we can assume that
the graph induced by our pattern is connected, since we can apply the
algorithm to each connected component separately. The connectivity implies
that the left endpoint of the rightmost pattern is at most $(n-1)w$.
Let $0<k\leq n$ be an integer and let $\cG(k)$ be the intersection graph
induced by the patterns with left endpoints in $\civ{0,kw}$. Let $\cI(k)$ be
the set of input patterns with left endpoints in $\roiv{(k-1)w, kw}$ and let
$S\subseteq\cI(k)$. Let $A(k,S)$ be the size of a minimum dominating set $D$
of $\cG(k)$ for which $D \cap \cI(k)=S$. By Lemma~\ref{lem:sparse_dom_set} it
follows that $|S|\leq 3w$.

The following recursion holds for $A(k,S)$ if we define $A(0,S):=0$:
\[
A(k,S)=\min\Big\lbrace  A(k-1,S') + |S| \;\Big| 
S' \subset \cI(k-1),\,
|S'|\leq 3w,\,
S\cup S' \text{ dominates } \cI(k) \Big\rbrace.
\]
\skb{The inequality ``$\leq$'' is easy to see, we are only minimizing over the sizes of feasible dominating sets of $\cG(k)$. For the other direction (``$\geq$''), Lemma~\ref{lem:sparse_dom_set} implies that there is a minimum dominating set containing at most $3w$ left endpoints from both $\cI(k-1)$ and $\cI(k)$, therefore its size is $A(k-1,S') + |S|$ for some $S' \subset \cI(k-1),\,|S'|\leq 3w$ that together with $S$ dominates $\cI(k)$.}
The number of subproblems for a fixed value of $k$ is $\sum_{j=0}^{3w}
\binom{n}{j} = O(n^{3w})$; thus the number of subproblems is $O(n^{3w+1})$.
Computing the value of a subproblem requires looking at $O(n^{3w+1})$
potential subsets $S'$, and $O(n^2)$ time is sufficient to check whether
$S\cup S'$ dominates $\cI(k)$. Overall, the running time of our algorithm is
$O(n^{6w+4})$.
\end{proof}

\begin{lemma}\label{lem:rational_point_pattern}
If $Q$ is a point pattern so that the distance ratios of any two point pairs of $Q$ are 
rational, then \textsc{$Q$-Intersection Dominating Set} can be solved in polynomial time.
\end{lemma}

\begin{proof}
By shifting and rescaling, we may assume without loss of generality that the leftmost
point in $Q$ is in the origin and that all points in $Q$ have integer coordinates. \skb{(Note that this could not be done if the pattern contained an irrational distance ratio.)}
We define a new pattern $Q'$ that results from $Q$ by replacing point $0$ by the 
interval $\civ{0,1/3}$.

Now consider an intersection graph whose vertices are associated with $x_i+Q$ where
$x_1\le x_2\le\cdots\le x_n$.
We assume without loss of generality that the graph is connected and that all $x_i$
are integers.
It can be seen that the intersection graph does not change, if every object $x_i+Q$ 
is replaced by the object $x_i+Q'$.
Since pattern $Q'$ contains the interval $\civ{0,1/3}$, we may simply apply
Lemma~\ref{lem:pattern_with_interval} to compute the optimal dominating set 
in polynomial time.
\end{proof}

\begin{lemma}\label{lem:irrational_dist_ratio_NPC}
If $Q$ is a point pattern that contains two point pairs with an irrational
distance ratio, then \textsc{$Q$-Intersection Dominating Set} is
\NP-complete.
\end{lemma}

\begin{proof}
The containment in \NP is trivial; we show the hardness by reducing from
dominating set on induced triangular grid graphs. (These are finite induced
subgraphs of the triangular grid, which is the graph with vertex set
$V=\Ints^2$ and edge set $E=\big\{\big((a,b),(a+\alpha,b+\beta)\big)\;:\;
|\alpha|\leq 1, |\beta|\leq 1, \alpha\neq \beta\big\}$.) The \NP-hardness of
dominating set in induced triangular grid graphs is proven in the appendix.
Note that the dominating set problem is known to be \NP-hard on induced grid
graphs, but this does not imply the hardness on triangular grids, because
triangular grid graphs are not a superclass of grid graphs.

We show that the infinite triangular grid can be realized as a
$Q$-intersection graph, where the $Q$-translates are in a bijection with the
vertices of the triangular grid. Therefore, any induced triangular
grid graph is realized as the intersection graph of the $Q$-translates
corresponding to its vertices.

Rescale $Q$ so that it has span $1$. It cannot happen that all the points are
rational, because it would make all distance ratios rational as well. Let
$x^*\in Q$ be the smallest irrational point. Let $a\in \Ints$, and consider
the intersection of the translate $ax^*+Q$ with the set $\Ints+Q$. We claim
that this intersection is non-empty only for a finite number of values $a\in
\Ints$. Suppose the opposite. Since $Q$ is a finite pattern, there must be a
pair $z,z' \in Q$ such that $ax^*+z=b+z'$ has infinitely many solutions
$(a,b)\in \Ints^2$. In particular, there are two solutions $(a_1,b_1)$ and
$(a_2,b_2)$ such that $a_1 \neq a_2$ and $b_1\neq b_2$. Subtracting the two
equations we get $(a_1-a_2)x^*=b_1-b_2$, which implies
$x^*=\frac{b_1-b_2}{a_1-a_2}$. This is a contradiction since $x$ is
irrational.

Let $y^*=a'x^*$, where $a'$ is the largest value $a$ for which $ax^*+Q$
intersects $\Ints+Q$. It follows that
$\big\{j\in \Ints \,\big|\, (jy^*+Q)\cap (\Ints + Q) \neq \emptyset \big\}
= \big\{{-1},0,1\big\}.$

Consider the intersection graph induced by the sets
$\big\{jy^* + k +Q \,\big|\, (j,k) \in \Ints^2\big\}$.
The above shows that a fixed translate $jy^*+k+Q$ is not intersected by the
translates $(j+\alpha)y^*+(k+\beta)+Q$ if $|\alpha|\ge2$. It is easy to see
that $|\beta|\ge2$ does not lead to an intersection either. Also note that
$\alpha=\beta=\pm 1$ does not give an intersection; however all the remaining
cases are intersecting, i.e., if
\[(\alpha,\beta)\in \big\{ (-1,0),(-1,1),(0,-1),(0,0),(0,1),(1,-1),(1,0)\big\}\]
then $(j+\alpha)y^*+(k+\beta)+Q$ intersects $jy^*+k+Q$. Thus, the intersection
graph induced by $\big\{jy^* + k + Q \,\big|\, (j,k) \in \Ints^2\big\}$ is a
triangular grid.
\end{proof}

% \begin{theorem}\label{thm:DSNPH}
% Dominating set is \NP-hard on (induced) triangular grid graphs.
% \end{theorem}

% \mypara{Proof sketch.}{ We do a reduction from the Planar 3SAT problem. Suppose that our planar 3CNF formula has $n$ variables and $m$ clauses. Each clause will be assigned a single vertex, and each variable will be represented by a cycle of length $O(m)$. In addition, the vertex-clause connections will be represented by grid paths, which we call \emph{literal paths}. The literal paths will branch off of distinct points of the variable cycle. Such a representation can be computed in polynomial time.

% In a minimum dominating set, only every third vertex needs to be selected from an induced path or an induced cycle. Each variable cycle has length that is a multiple of $3$, and we can add an ear so that a dominating set must select vertices whose index within the cycle is congruent to $0$ or $1$ modulo $3$. On the literal paths (the paths connecting a variable cycle's vertex to a clause vertex) we have the same kind of binary choice, the two choices corresponding to the truth value of the literal. We can make sure that paths corresponding to positive and negative literals have the correct length modulo $3$, so that the clause vertex is dominated by the last dominating point of a literal path if and only if at least one of the literals is true.}

\begin{lemma}\label{lem:irrational_dist_ratio_FPT}
If $Q$ is a point pattern that has point pairs with an irrational
distance ratio, then \textsc{$Q$-Intersection Dominating Set} has an \FPT
algorithm parameterized by solution size.
\end{lemma}

\begin{proof} \skb{In polynomial time, we can remove all duplicate translates,
since a minimum dominating set contains at most one of these objects, and any
minimum dominating set of the resulting graph is a dominating set of the
original graph. Suppose our pattern consists of $t$ points. In the
duplicate-free graph, point $i$ of the pattern translate may intersect point
$j$ of another translate, for some $i\neq j$, so the maximum degree is
$t^2-t$. Therefore} we are looking for a dominating set in a graph of bounded
degree. Hence, a straightforward branching approach gives an \FPT algorithm:
choose any undominated vertex $v$; either $v$ or one of its at most \skb{$t^2-t$}
neighbors is in the dominating set, so we can branch \skb{$t^2-t+1$ ways. If all
vertices are dominated} after choosing $k$ vertices, then we have found a
solution. This branching algorithm has depth $k$, with linear time required at
each branching, so the total running time is $O\big(t^{2k}(|V|+|E|)\big)$.
\end{proof}

In the following lemma we show that any graph can be
obtained as a  1-dimensional pattern intersection graph for a suitable choice
of pattern $Q$. Consequently \textsc{$Q$-Intersection Dominating Set} is
\Wtwo-complete if the pattern $Q$  is part of the input.

\begin{lemma}\label{lem:large_pattern_proof}
Let $G$ be a graph with vertex set $V=\{v_1,\ldots,v_n\}$ and edge set $E=\{e_1,\ldots,e_m\}$.
Then there exists a finite pattern $Q\subseteq\mathbb{R}$ and there exist real numbers $x_1,\ldots,x_n$
so that $\{v_i,v_j\}\in E$ if and only if $(x_i+Q)\cap(x_j\cap Q)\ne\emptyset$.
\end{lemma}

\begin{proof}
Define $q=2(n+m)$.
For every edge $e_k$, we let $a(k)$ and $b(k)$ with $a(k)<b(k)$ denote the indices of
its incident vertices.
For $k=1,\ldots,m$, the pattern $Q$ contains the two integers
\[ 4^{q+k}-4^{a(k)} \mbox{\qquad and\qquad} 4^{q+k}-4^{b(k)}. \]
Furthermore define $x_i=4^i$ for $i=1,\ldots,n$.

First suppose $e_k=\{v_i,v_j\}\in E$ with $i=a(k)$ and $j=b(k)$.
Then $x_i+Q$ and $x_j+Q$ both contain the number $4^{q+k}$, so that indeed
$(x_i+Q)\cap(x_j+ Q)\ne\emptyset$.
Next suppose $(x_i+Q)\cap(x_j+ Q)\ne\emptyset$.
This means that there exist edges $e_k$ and $e_{\ell}$ and $c\in\{a(k),b(k)\}$ and
$d\in\{a(\ell),b(\ell)\}$ so that
\[ 4^i+(4^{q+k}-4^{c}) ~=~ 4^j+(4^{q+\ell}-4^{d}). \]
Since the exponents $q+k$ and $q+\ell$ are much larger than the other exponents in this equation,
they must coincide with $k=\ell$.
Without loss of generality, this leads to $c=a(k)$ and $d=b(k)$.
The equation boils down to $4^i-4^{a(k)}=4^j-4^{b(k)}$, which implies $i=a(k)$ and $j=b(k)$.
Hence vertices $v_i$ and $v_j$ are indeed connected by an edge $e_k$.
\end{proof}

\begin{remark}
In our handling of the problem, the pattern was part of the problem definition. Making the pattern part of the input leads to an \NP-complete problem: Lemma~\ref{lem:irrational_dist_ratio_NPC} can be adapted to this scenario. If we also allow the size of the pattern to depend on the input, then the problem is \Wtwo complete (when parameterized by solution size) by Lemma~\ref{lem:large_pattern_proof}. 

We propose the following problem for further study, where the pattern depends on the input, but has fixed size.
\end{remark}

\mypara{Open question.}{ Let $Q$ be the pattern defined by two unit intervals on a line at distance $\ell$. Is there an \FPT algorithm (either with parameter $k$ or $k+\ell$) on intersection graphs defined by translates of $Q$, that can decide if such a graph has a dominating set of size $k$?
It can be shown that this problem is \NP-complete, and 
Theorem~\ref{thm:genW1containment} below shows that it is contained in \Wone.}

\section{Higher dimensional shapes: \Wone vs. \Wtwo}

In this section we show that dominating set on intersection graphs of 2-dimensional objects is contained in \Wone if the shapes have a
constant size description. First, we demonstrate the method on unit disk
graphs, and later we state a much more general version where the shapes are
semi-algebraic sets. In order to show
containment, it is sufficient to give a non-deterministic algorithm that has
an FPT time deterministic preprocessing, then a nondeterministic phase where
the number of steps is only dependent on the parameter. More precisely, we
use the following theorem.

\begin{theorem}[\cite{Flum06}]\label{thm:w1hdesc}
A parameterized problem is in \Wone if and only if it can be computed by a
nondeterministic RAM program accepting the input that
\begin{enumerate}
\item performs at most $f(k)p(n)$ deterministic steps;
\item uses at most $f(k)p(n)$ registers;
\item contains numbers smaller than $f(k)p(n)$ in any register at any time;
\item for any run on any input, the nondeterministic steps are among the last
$g(k)$ steps.
\end{enumerate}
Here $n$ is the size of the input, $k$ is the parameter, $p$ is a polynomial
and $f,g$ are computable functions. The non-deterministic instruction is
defined as guessing a natural number between 0 and the value stored in the
first register, and storing it in the first register. Acceptance of an input is defined as having a computation path that accepts.
\end{theorem}

\begin{theorem}\label{thm:dsudg_in_w1}
The dominating set problem on unit disk graphs is contained in \Wone.
\end{theorem}

\begin{figure}
\begin{center}
\includegraphics[scale=0.6]{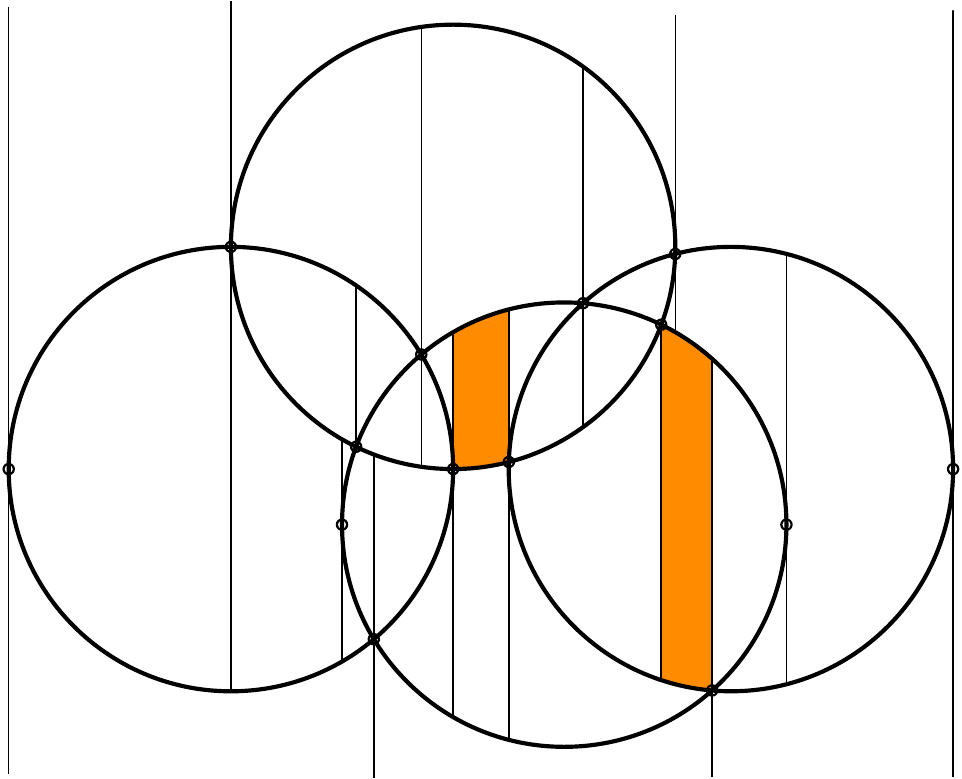}
\end{center}
\caption{Two faces of a vertical decomposition.} \label{fig:cells}
\end{figure}

\begin{proof}
Let $P$ be the set of centers of the  unit disks that form the input instance. For a subset $D\subseteq P$, let $\cC_2(D)$ and $\cD_2(D)$ be the set of circles and disks of radius 2, respectively,  centered at the points of $D$.
(Note that $D$ is a dominating set if and only if $\bigcup\cD_2(D)$, the union of the disks in $\cD_2(D)$, covers all points in $P$.)
Shoot a vertical ray up and down from each of the $O(k^2)$ intersection points between the circles of $\cC_2(D)$, and also from the leftmost and rightmost point of each circle. Each ray is continued until it hits a circle (or to infinity). The arrangement we get is a \emph{vertical decomposition}~\cite{Berg08} (see Fig.~\ref{fig:cells}). Each face of this decomposition is defined by at most four circles. This is not only true for the 2-dimensional faces, but also for the 1-dimensional faces (the edges of the arrangement) and 0-dimensional faces (the vertices). We consider the faces to be relatively open, so that they are pairwise disjoint.

In our preprocessing phase, we compute all potential faces of a vertical decomposition of any subset $D\subseteq P$ by looking at all 4-tuples of circles from $\cC_2(P)$. We create a lookup table that contains the number of input points covered by each potential face in $O(n^4)$ time.

Next, using nondeterminism we guess $k$ integers, representing the points of
our solution; let $D$ be this point set. The rest of the algorithm
deterministically checks if $D$ is dominating. We need to compute the
vertical decomposition of $\cC_2(D)$; this can be done in $O(k^2)$
time~\cite{Berg08}. Finally, for each of the $O(k^2)$ resulting faces
of $\bigcup\cD_2(D)$, we can get the number of input points covered
from the lookup table  in constant time. We accept if these numbers sum to
$n$.  By Theorem~\ref{thm:w1hdesc} we can thus conclude that 
dominating set on unit disk graphs is in \Wone.
\end{proof}

%---------------------------------------------------------

In order to state the general version of this theorem, we introduce
semi-algebraic sets. A \emph{semi-algebraic set} is a subset of $\Reals^d$
obtained from a finite number of sets of the form $\{x \in \Reals^d \;|\; g(x)
\geq 0\}$, where $g$ is a $d$-variate polynomial with integer coefficients,
by Boolean operations (unions, intersections, and complementations). Let
$\Gamma_{d,\Delta,s}$ denote the family of all semi-algebraic sets in
$\Reals^d$ defined by at most $s$ polynomial inequalities of degree at most
$\Delta$ each. If $d, \Delta, s$ are all constants, we refer to the sets in
$\Gamma_{d,\Delta,s}$ as constant-complexity semi-algebraic sets.

Let $\cF$ be a family of constant complexity semi-algebraic sets in
$\Reals^d$ that can be specified using $t$ parameters $a_1,\ldots,a_t$. If
the expressions defining $\cF$ are also polynomials in terms of the
parameters, then we call $\cF$ a \emph{$t$-parameterized family of semi-
algebraic sets}. For example, the family of all balls in the $\Reals^3$ is a
4-parameterized family of semi- algebraic sets, since any ball can be
specified using an inequality of the form
$(x_1-a_1)^2+(x_2-a_2)^2+(x_3-a_3)^2 - a_4^2\leq 0$. As another example, the
family of all triangles in the plane is a 6-parameterized algebraic set,
since any triangle is the intersection of three half-planes, and any half-
plane can be specified using two parameters.

\begin{theorem}\label{thm:genW1containment}
Let $\cF$ be a $t$-parameterized family of semi-algebraic sets, for some constant~$t$.
Then dominating set is in \Wone for intersection graphs defined by~$\cF$.
\end{theorem}

The proof is very similar to the proof of the special case of unit disks. We give a proof sketch before introducing the machinery required for the full proof.

\begin{proof}[Proof sketch of Theorem~\ref{thm:genW1containment}]
By definition, any set $S\in\cF$ can be specified using $t$ parameters $a_1,\ldots,a_t$. Thus we can represent $S$ by the point $\bp(S) := (a_1,\ldots,a_t)$ in $\Reals^t$. 
Conversely, for a point $(a_1,\ldots,a_t) \in \Reals^t$, let $S(a_1,\ldots,a_t)$ be the corresponding semi-algebraic set. Now we define, for any set $S\in\cF$, a region $\bR(S)$ as follows:
\[
\bR(S) := \{ (a_1,\ldots,a_t) \in \Reals^t : S(a_1,\ldots,a_t) \cap S \neq \emptyset \}.
\]
Thus for any two sets $S_1,S_2\in\cF$ we have that $S_1\cap S_2\neq\emptyset$
if and only if $\bp(S_1) \in \bR(S_2)$. 

Now consider a set $\cS \subset \cF$ of $n$ sets from the family~$\cF$. We
proceed in a similar way as in the proof of Theorem~\ref{thm:dsudg_in_w1},
where the sets $\bR(S)$ for $S\in\cS$ play the same role as the radius-2
disks in that proof. Consider any subset $\cD\subseteq \cS$, and note that
$\cD$ is a dominating set if and only if $\bigcup_{S\in \cD} \bR(S)$ contains
the point set $\{ \bp(S) | S\in \cS\}$.

Now we can decompose the arrangement defined by $\{ \bR(S) : S\in \cD\}$ into
polynomially many cells using a so-called \emph{cylindrical
decomposition}~\cite{Arnon84}; note that such a decomposition is made
possible by the fact that the regions $\bR(S)$ are semi-algebraic. (This decomposition plays the role of the vertical decomposition in
the proof for unit disks.) Each cell of the cylindrical decomposition is
defined by at most $t'$ regions $\bR(S)$, for some $t'=O(1)$. Thus, for each
subset of at most $t'$ regions $\bR(S)$, we compute all cells that arise in
the cylindrical decomposition of the subset. The number of possible cells is polynomial in $n$.

In the preprocessing phase, we compute for each possible cell the number of
points $\bp(S)$ contained in it, and store the results in a lookup table. The
next phase of the algorithm is the same as for unit disks: we guess a
solution, compute the cells in the cylindrical decomposition of the
corresponding arrangement, and check using the lookup table if the guessed
solution is a dominating set.
\end{proof}

\begin{definition}[First order formula]
%http://www.kurims.kyoto-u.ac.jp/~kyodo/kokyuroku/contents/pdf/1764-05.pdf
A first-order formula (of the language of ordered fields with
parameters in $\mathbb{R}$) is a formula that can be constructed according to
the following rules:
\begin{enumerate}
\item  If $Q\in \Reals[X_{1}, \cdots, X_{d}]$, then $Q\star 0$, where
$\star\in\{=, >, <\}$, is a formula.
\item If $\phi$ and $\psi$ are formulas, then their conjunction
$\phi\wedge\psi$, their disjunction $\phi\vee\psi$, and the negation
$\neg\phi$ are formulas.
\item  If $\phi$ is a formula and $x$ is a variable ranging over $\mathbb{R}$, then $(\exists x) \phi$ and $(\forall x) \phi$ are formulas
\end{enumerate}
\end{definition}

A formula that can be obtained by using only the first two of the above steps
is called \emph{quantifier-free}. The \emph{realization} of a formula $\Phi$
with $t$ free variables is $\reali(\Phi)=\{x\in \Reals^t
\;|\;\Phi(x)=\true\}$. A \emph{semi-algebraic set} is the realization of a
quantifier-free first-order formula. Our definition of a semi-algebraic family is equivalent to saying that it is defined by a first order formula $\Phi(a,x)$ in the following sense: let $\Psi_a(x)=\Phi(a,x)$ for all $a\in \Reals^t$. Then the sets in the family are $S(a)=S(a_1,\dots,a_t)=\reali(\Psi_a(x))\quad (a\in \Reals^t)$. Note that $\Phi$ must have constant complexity, i.e., the degree of the polynomials, the number of inequalities and the number of variables is a constant.

\begin{proof}[Proof of Theorem~\ref{thm:genW1containment}]
Throughout this proof, $h_1,h_2,\dots$ denote computable functions.
Consider a quantifier-free first order formula $\Phi(a,x)$ that defines our $t$-parameterized family. The condition $S(a) \cap S(a') \neq \emptyset$ is equivalent to
\[\Pi(a,a') = \Big((\exists x):\Phi(a,x) \;\wedge \; \Phi(a',x)\Big).\]
We can use quantifier elimination~\cite{Basu2006} on the formula $\Pi(a,a')$ to gain an equivalent quantifier-free formula $\Xi(a,a')$, where the maximum degree $\Delta$ and the number of inequalities $s$ are at most singly exponential in $t$. Let $h_1$ be a singly exponential function for which $h_1(t)\geq \max\{2t,\Delta,s\}$. For any $a\in \Reals^{t}$, let $\Theta_a(a')=\Xi(a,a')$.

Consider an intersection graph corresponding to a finite set of parameter values $A\subset \Reals^{t}$. Let $V$  be the vertex set of our intersection graph: $V = \{S(a) \;|\; a \in A\}$. Furthermore, let $\bR$ be the function that assigns any $S(a) \in V$ the shape $\bR(S(a))=\reali(\Theta_a)$.

In the verification phase, we will guess a dominating vertex set $\cD$, and we are going to apply a cylindrical algebraic decomposition to the semi-algebraic sets $\bR(\cD)=\{\bR(S)\;|\; S\in D\}$. Each semi-algebraic set in $\bR(\cD)$ can have at most $h_2(t)$ connected components.

Every cell in the cylindrical decomposition can be defined by a tuple of connected components, and the tuple size depends only on the dimension and the degree of polynomials used for our semi-algebraic sets. In our case, the dimension is $t$ and the degrees are at most $h_1(t)$, therefore the tuple size can be upper bounded by a function of $t$, let it be $h_3(t)$.

Our algorithm is as follows. In the preprocessing phase, we enumerate all possible cells in $(nh_2(t))^{h_3(t)}=poly(n)$ time, and in each cell in $O(n)$ time we count the number of points covered from $A$, and save the information in a lookup table.

Next, we make the $k$ guesses, that correspond to the vertex identifiers of the dominating set $\cD$. We create the cylindrical algebraic decomposition for $\bR(S) \quad \{S\in \cD\}$. For each cell covered by $\bigcup_{S\in \cD} \bR(S)$, we sum the entries from the lookup table. We accept if the result is $n$. Note that the guesses, the decomposition and the lookup together take $h_4(k)$ time.
\end{proof}

\mypara{W[1]-hardness for simple polygon translates.}{ We generalize a proof by Marx~\cite{Marx06} for the \Wone-hardness of dominating set in unit square/unit disk graphs. Our result is based on the observation that many 2-dimensional shapes share the crucial properties of unit squares when it comes to the type of intersections needed for this specific construction. We prove the following theorem.}

\begin{theorem}
The dominating set problem is \Wone-hard for intersection graphs of the
translates of a simple polygon in $\Reals^2$.
\end{theorem}

Our proof uses the same global strategy as Marx's proof~\cite{Marx06} for the \Wone-hardness of dominating set for intersection graphs of squares. (We give an overview of the proof in Appendix~\ref{app:Marx_generic}.) To apply this proof strategy, all we need to prove is that the family of shapes for which we want to prove \Wone-hardness has a certain property, as defined next.

We say that a shape $S\subseteq \Reals^2$ is
\emph{square-like} if there are two base vectors $\bb_1$ and $\bb_2$ and for
any $n$ there are two small offset vectors $\bu_1=\bu_1(n)$ and $\bu_2=\bu_2(n)$ with the following properties. Define $S(i,j) := S + i\bu_1 + j\bu_2$ for all $-n^2\leq i,j \leq n^2$, and consider the set $\cK := \{S(i,j) : -n^2\leq i,j\leq n^2\}$. Note that $\cK$ consists of $(2n^2+1)^2$ translated copies of $S$ whose reference points from a $(2n^2+1)\times(2n^2+1)$ grid. Also note that $S = S(0,0)$. For the shape $S$ to be square-like, we require the following properties:
\begin{itemize}
\item $\cK$ is a clique in the intersection graph, i.e.,\\
\[\text{for all} -n^2 \!\leq \!i,\!j \!\leq\! n^2 \text{ we have:} \; S\cap S(i,j) \neq \emptyset.\]
\item ``Horizontal'' neighbors intersect only when close:\\
\[\text{for all} -n^2 \!\leq\! j\! \leq\! n^2 \text{ we have:} \; S\cap (\bb_1 + S(i,j)) \neq \emptyset \iff i \leq 0.\]
\item ``Vertical'' neighbors intersect only when close: \\
\[\text{for all} -n^2 \!\leq\! i \!\leq\! n^2 \text{ we have:} \; S\cap (\bb_2 + S(i,j)) \neq \emptyset \iff j \leq 0.\]
\item Distant copies of $\cK$ are disjoint:\\
\[\text{for all} -n^2 \!\leq\! i,\!j,\!i'\!,\!j'\!\leq\! n^2 \text{ we have:} \;|k|+|\ell|\geq 2 \Rightarrow S(i,j)\cap (k\bb_1 + \ell\bb_2 + S(i',j')) = \emptyset.\]
\end{itemize}
Moreover, we require that each of the vectors can be represented on $O(\log n)$ bits. It is helpful to visualize a square grid, with unit side lengths $\bb_1$ and $\bb_2$, where we place the centers of unit squares with small offsets compared to the grid points. We are requiring a very similar intersection structure here. See Figure~\ref{fig:pgrid} for an example of a good choice of vectors.

\begin{figure}
\begin{center}
\includegraphics[scale=0.8]{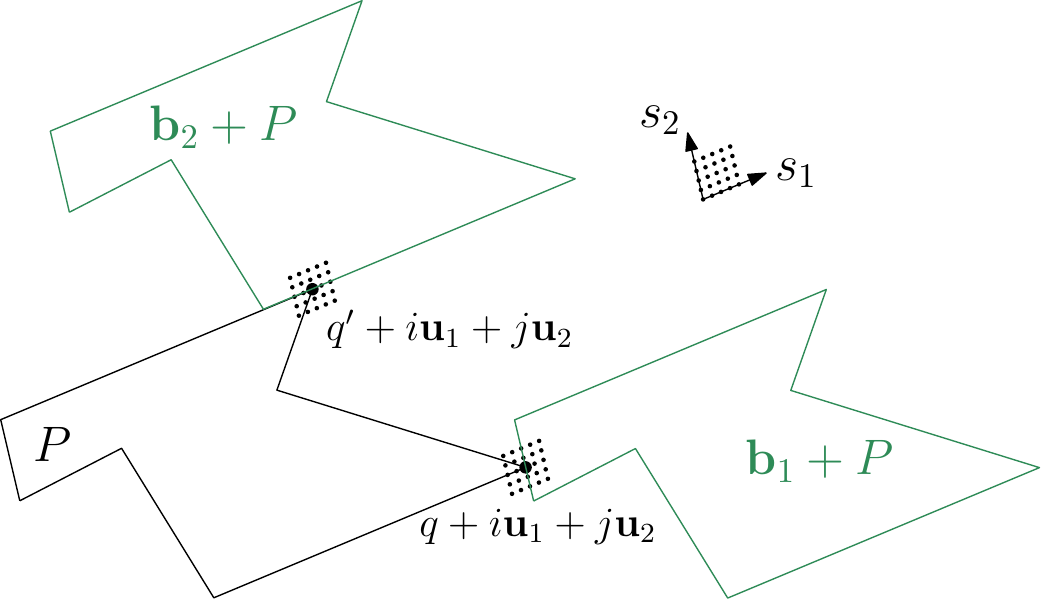}
\end{center}
\caption{A good choice of $\bb_1,\bb_2$ and offsets.}
\label{fig:parallelogram}
\end{figure}

Since the above properties are sufficient for the construction given by
Marx~\cite{Marx06}, we only need to prove the following theorem.

\begin{theorem}\label{thm:allissquare}
Every simple polygon is square-like.
\end{theorem}

We give a short overview of the proof technique. First, we would like to
define a ``horizontal'' direction, i.e., a good vector $\bb_1$. A natural
choice would be to select a diameter of the polygon (see $b_0$ in Figure~\ref{fig:vecfind}), however that would result in $S$ and $\bb_1+S$
intersecting each other at vertices. That would pose a severe restriction on
the offset vectors; therefore, we use a perturbed version of a diameter,
making sure that the intersection of $S$ and $\bb_1+S$ is realized by a
polygon side from at least one party. The direction of this polygon side also
defines a suitable direction of the offset vector $\bu_2$: because of the second
property, choosing $\bu_2$ to be parallel to this direction ensures the
independence with respect to the choice of $j$.

Next, we define the other base vector $\bb_2$. This definition is based on
laying out an infinite sequence of translates horizontally next to each other
(Figure~\ref{fig:secvecfind}). We want a translate of this sequence to
touch the original sequence in a ``non-intrusive'' way: small perturbations of
$\bb_2+S$ should only intersect $S$, but stay disjoint from $\bb_1+S$ or
$-\bb_1+S$. This is fairly easy to achieve; again with a small perturbation of
our first candidate vector we can also ensure that the intersection between
$\bb_2+S$ and $S$ is not a vertex-vertex intersection. Finally, a suitable direction
for the offset vector $\bu_1$ is given by the polygon side taking part in the
intersection between $\bb_2+S$ and $S$.

We now give the formal proof of Theorem~\ref{thm:allissquare}.

\begin{figure}
\begin{center}
\includegraphics[scale=0.7]{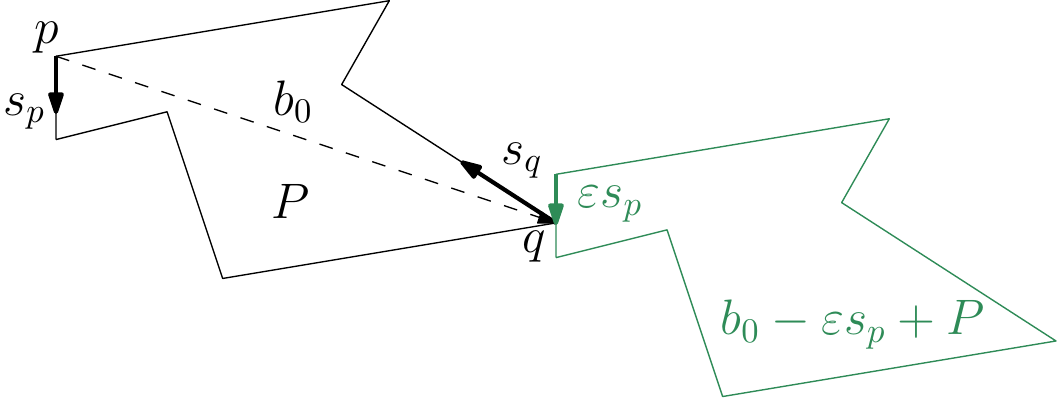}
\end{center}
\caption{Defining $\bb_1$.}
\label{fig:vecfind}
\end{figure}

\begin{figure}
\begin{center}
\includegraphics[scale=0.7]{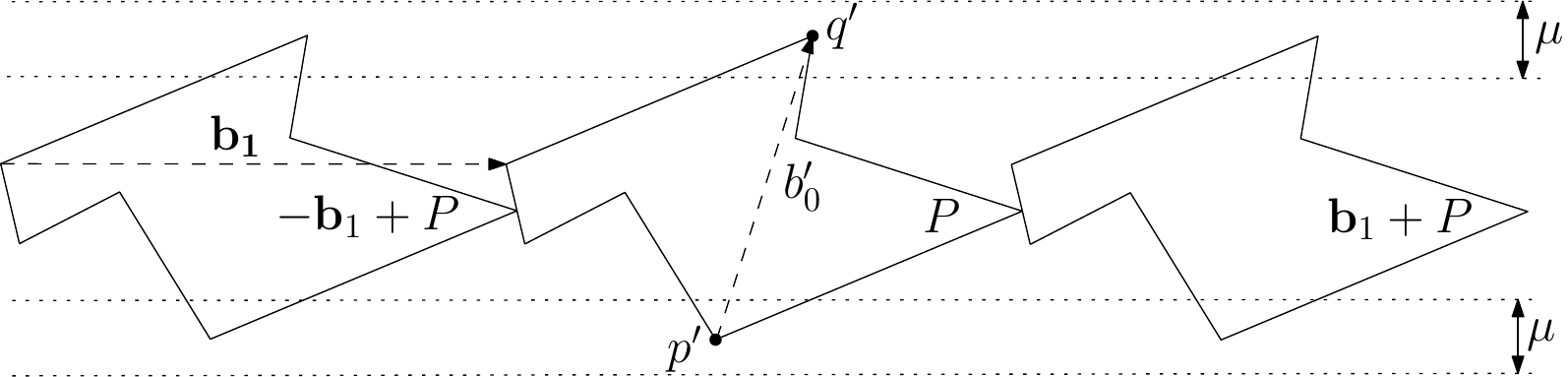}
\end{center}
\caption{Defining $\bb_2$.}
\label{fig:secvecfind}
\end{figure}

% \begin{figure}
% \begin{center}

% \end{center}
% \caption{Defining $\bb_2$.}
% \label{fig:secondvecfind}
% \end{figure}

\begin{proof}

Let $P$ be a simple polygon, and let $p$ and $q$ be two endpoints of a
diameter of $P$. Let $b_0=q-p$. Since $P$ is a polygon, both $p$ and $q$ are
vertices. Let $s_p$ and  $s_q$ be unit vectors in the direction of the side of
$P$ that follows vertex $p$ and $q$ in the counter- clockwise order. Let
$\eps>0$ be a small number to be specified later. Consider the intersection of
$P$ and the translate $b_0+\eps s_q+P$. If $\eps$ is small enough, then
depending on the angle of $s_p$ and $s_q$, this intersection is either the
point $b_0+\eps s_q$, or part of the side with direction $s_q$, or it is an
intersection of positive area. Figure~\ref{fig:vecfind}
illustrates the third case. In the first case, let $\bb_1=b_0+\eps s_p$; in
the second and third case, let $\bb_1=b_0-\eps s_q$. Furthermore, let
$s_1=\bb_1-b_0$. We will later use $s_1$ to define the offset vector $\bu_2$.

Imagine that $\bb_1$ is the horizontal direction, and consider the set
$P_{\infty}=\{k\bb_1+P\;|\;k \in \mathbb{Z}\}$ (Figure~\ref{fig:secvecfind}). Its top and bottom boundary are infinite periodic
polylines, with period length $|\bb_1|$. Take a pair of horizontal lines that
touch the top and bottom boundary. By manipulating $\eps$ in the definition of
$\bb_1$, we can achieve a general position in the sense that both of these
lines touch the respective boundaries exactly once in each period, moreover,
there is a value $\mu$, such that there are no vertices other than the
touching points in the $\frac{\mu}{2}$-neighborhood of the touching lines. Let
$p'$ and $q'$ be vertices touched by the bottom and top lines inside $P$, and
let $b'_0=q'-p'$. Similarly as before, the direction of the sides following
$p'$ and $q'$ counter-clockwise are denoted by $s_{p'}$ and $s_{q'}$. If the
intersection of $P$ and the translate $b'_0+\eps s_{q'}+P$ has zero area, then
let $\bb_2=b'_0-\eps s_{q'}$; otherwise, (if the area of the intersection is
positive), let $\bb_2=b'_0+\eps s_{p'}$. We denote by $s_2$ the difference
$\bb_2-b'_0$. If $s_2$ and $s_1$ are parallel, then we can define $s_2$
similarly, by replacing the sides $s_{p'}$ and $s_{q'}$ with the sides that
follow $p'$ and $q'$ in clockwise direction. The new direction of $s_2$ will
not be parallel to the old one, therefore it will not be parallel to $s_1$.

\begin{figure}
\begin{center}
\includegraphics[scale=0.7]{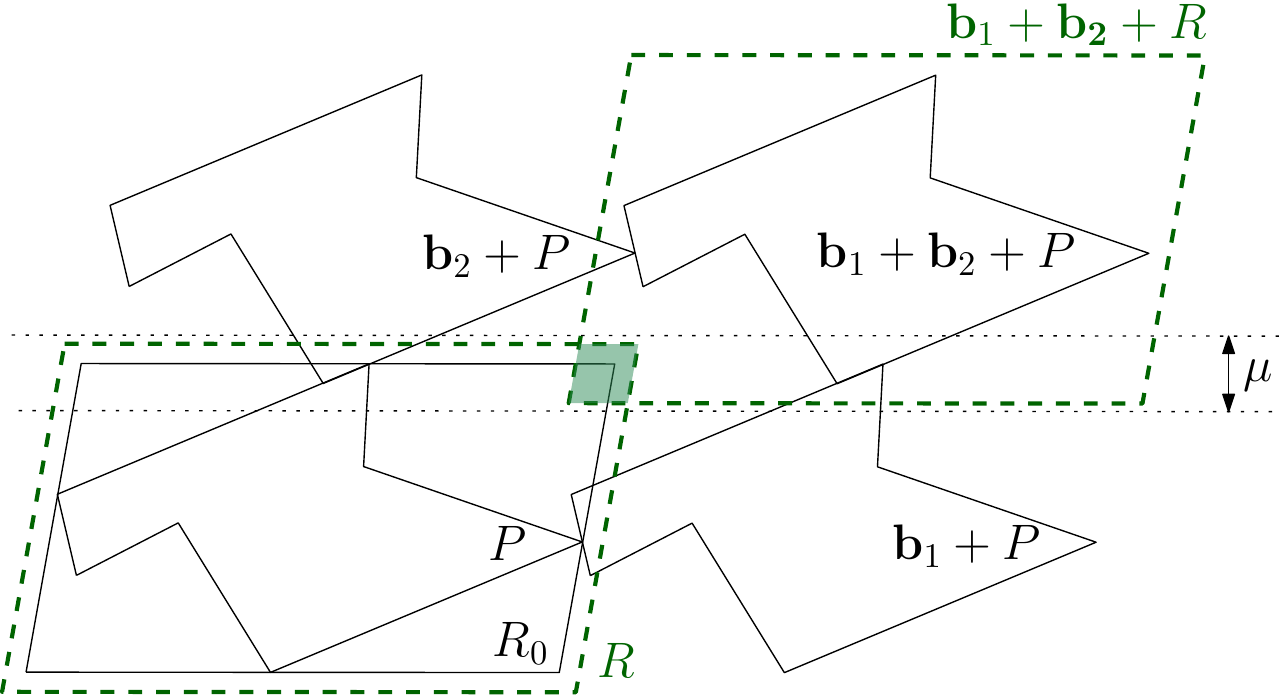}
\end{center}
\caption{Part of the grid $k\bb_1+\ell\bb_2+P$.}
\label{fig:pgrid}
\end{figure}

We need to choose the values of $\bu_1$ and  $\bu_2$. Let $\bu_1 =
\frac{\eps}{2n^2} s_2$ and let $\bu_2 = \frac{\eps}{2n^2}
s_1$. We claim that if $\eps$ is small enough, then $P$ is
square-like for the vectors $\bb_1, \bb_2, \bu_1, \bu_2$. It is easy to check
that for a small enough value of $\eps$, the first condition is satisfied,
namely that $P\cap P(i,j) \neq \emptyset \; \text{ for all }
-n^2 \leq i,j \leq n^2$.

Next, we show that for $i\leq 0$, the intersection of $P$ and $\bb_1 +
P(i,j)$ is non-empty. Consider the small grid of points $q -
i\bu_1 - j\bu_2, \; -n^2\leq i,j\leq n^2$ (see
Figure~\ref{fig:parallelogram}). This grid fits into a parallelogram whose
sides are parallel to $s_2$ and $s_1$. Notice that if $\eps$ is small enough,
then $q - i\bu_1 - j\bu_2$ for all $-n^2 \leq i <0$ and $-n^2\leq j \leq n^2$ is
contained in $\bb_1+P$, thus the intersection $P\cap (\bb_1 + i\bu_1 + j\bu_2
+ P)$ is non-empty if $i\leq 0$. Moreover, (if $\eps$ is small enough), then
no other type of intersection can happen by moving $\bb_1+P$ slightly: the
only sides that can intersect $\bb_1 + P(i,j)$ from $P$ are 
adjacent to $q$. Therefore, if $q$ is outside $\bb_1 + P(i,j)$,
then the intersection is empty --- which is true for $i>0$. A similar
argument works for the intersection of $P$ and $\bb_2 + P(i,j)$.

Let $R_0$ be a minimum area parallelogram containing $P$ whose sides are
parallel to $\bb_1$ and $\bb_2$ (see Figure~\ref{fig:pgrid}). Notice that the
side lengths of this parallelogram are at most $|\bb_1|+\eps$ and
$|\bb_2|+\eps$ respectively. Let $\bar{P}=\bigcup \cK = \bigcup_{-n^2\leq i,j
\leq n^2} P(i,j)$. Notice that $\bar{P}$ is contained in the slightly larger
rectangle $R$ that we get by extending all sides of $R_0$ by $2\eps$.

Now consider the rectangle translates $k\bb_1+\ell\bb_2+R$. Since $\eps$ is
small enough, if either $k$ or $\ell$ is at least two then $R \cap (k\bb_1 +
\ell\bb_2 + R) = \emptyset$, so specifically, $\bar{P}$ is disjoint from
$k\bb_1 + \ell\bb_2 + \bar{P}$. It remains to show that $\bar{P}$ is disjoint
from $k\bb_1 + \ell\bb_2 + \bar{P}$ if $|k|=|\ell|=1$. Consider $\bar{P}$ and
$\bb_1+\bb_2+\bar{P}$ for example. They could only intersect inside $R \cap
(\bb_1+\bb_2+R)$; however, if $\eps < \frac{\mu}{4}$, then this is contained
in the $\mu$ wide horizontal strip defined earlier. By the definition of this
strip, it also means that there is an intersection point $q$ that is within
distance $O(\eps)$ from both $q'$ and $\bb_1+\bb_2+p'$. This would mean that
$|\bb_1|=O(\eps)$, and thus it can be avoided by choosing a small enough $\eps$.

Finally, we note that all restrictions on the value of $\eps$ are dependent
on the polygon $P$ itself, thus the length of the short vectors $\bu_1$ and
$\bu_2$ is $\Omega(n^{-2})$, and a precision of $O(n^{-2})$ is sufficient for
all the vectors, thus the vectors can be represented on $O(\log n)$ bits.
\end{proof}

We remark that it is fairly easy to further generalize the above theorem to
other families of objects, we can allow objects with certain curved boundaries
for example. A simple example of an object that is not square-like is a pair
of perpendicular disjoint unit segments: for any choice of offset vectors, the
set $\cK$ does not form a clique (as required by the first property of
square-like objects).

\mypara{\Wtwo-hardness for convex polygons.} {We conclude with the following
hardness result; the reduction uses a basic geometric idea that has been used
for hardness proofs before~\cite{Har09,Marx15}. Note the crucial difference
between the setting in this theorem, where the polygons defining the
intersection graph can be different and have description complexity dependent
on $n$, versus the previous settings (where we  had constant description
complexity and some uniformity among the object descriptions).}

\begin{theorem}
The dominating set problem is \Wtwo-hard for intersection graphs of convex
polygons.
\end{theorem}

\begin{proof}
A \emph{split graph} is a graph that has a vertex set which can be
partitioned into a clique $C$ and an independent set $I$. It was shown by
Raman and Saurabh~\cite{Raman2008} that dominating set is \Wtwo-hard on split graphs.
Thus it is sufficient to show that any split graph can be represented as the
intersection graph of convex polygons.

Let $G=(C\cup I,E)$ be an arbitrary split graph. Let $Q'$ be a regular
$2|I|$-gon and let $Q$ be the regular $I$-gon defined by every second
vertex of $Q'$. Notice that $Q' \setminus Q$ consists of small triangles, any
subset of which together with $Q$ forms a convex polygon.

The polygons corresponding to $I$ are small equilateral triangles, placed in
the interior of each small triangle of $Q'\setminus Q$. The polygon
corresponding to a vertex $v\in C$ whose neighborhood in $I$ is $N_I(v)$ is
the union of $Q$ and the small triangles corresponding to the vertices of
$N_I(v)$.

In this construction, the polygons corresponding to $C$ all intersect (they
all contain $Q$), and the polygons corresponding to $I$ are all disjoint. Finally, for any pair of vertices $u\in C$ and $v \in I$ the polygon of $u$ contains the polygon of $v$ if and only if $uv\in E$.
\end{proof}

\section{Conclusion}

We have classified the parameterized complexity of dominating set in intersection graphs defined by sets of various types in $\Reals^1$ and $\Reals^2$. More precisely, in $\Reals^1$, we gave a classification for the case when the intersection graph is defined by the translates of a fixed pattern that consists of points and intervals that is independent of the input. In $\Reals^2$, we have identified a fairly large class of \Wone-complete instances, namely, if our intersection graph is defined by a subset of a constant description complexity family of semi-algebraic sets. Even though our results hold for a large class of geometric intersection graphs, there are still some open problems. In particular, the complexity of dominating set on the following types of intersections graphs is unknown.

\begin{itemize}
\item translates of a 1-dimensional pattern that contains two unit intervals at some distance $\ell$ (given by the input) (\FPT vs. \Wone?)
\item translates of a 2-dimensional pattern that contains two disjoint perpendicular unit intervals (\FPT vs. \Wone?)
\item $n$ translates of a regular $n$-gon (\Wone vs. \Wtwo?)
\end{itemize}

\bibliography{new_w_cont}

\newpage

\appendix

\section{Dominating Set in the triangular grid}\label{app:triangular_grid}

\begin{theorem}\label{thm:DSNPH}
Dominating set is \NP-hard on (induced) triangular grid graphs.
\end{theorem}

We do a reduction from \textsc{Planar 3-SAT}. Consider a formula of $n$
variables and $m$ clauses, and let $G$ be the graph associated with the
formula: the vertices are the clauses and the variables, and the edges connect
a variable and a clause if the clause contains the variable.

For each variable $x$, we introduce a cycle of length $m+1$, where we put the
outgoing edges from $x$ consecutively on neighboring vertices of the cycle in
the same cyclic order as defined by a fixed planar drawing of the planar
graph $G$.
%See Figure~\ref{fig:variable_circle} for an example.
We obtain a new
planar graph $G'$ this way, which has maximum degree $3$. We now consider a
triangular grid drawing of $G'$, where the original vertices of $G'$ are
assigned to grid points, and the edges are replaced with grid paths. There is
such a drawing onto a grid of polynomial size, and it can be computed in
polynomial time ~\cite{Kant92}.

The edges are either edges of a \emph{variable cycle}, or they are on a path
from a variable cycle to a clause, which we will call \emph{literal path}s.

We scale this drawing by a factor of five and do some local modifications (that we describe below) in
order to make this an induced triangular grid graph. At each vertex or bend
where there are edges used whose angle is $\frac{\pi}{3}$, we modify the
surrounding area as seen in Figure~\ref{fig:resol_inc}. It is easy to see
that one can reroute the paths around a vertex or around a bend in a larger
hexagon in all the other cases.

\begin{figure}
\begin{center}
\includegraphics[trim={1cm 5.5cm 1cm 3.5cm},clip]{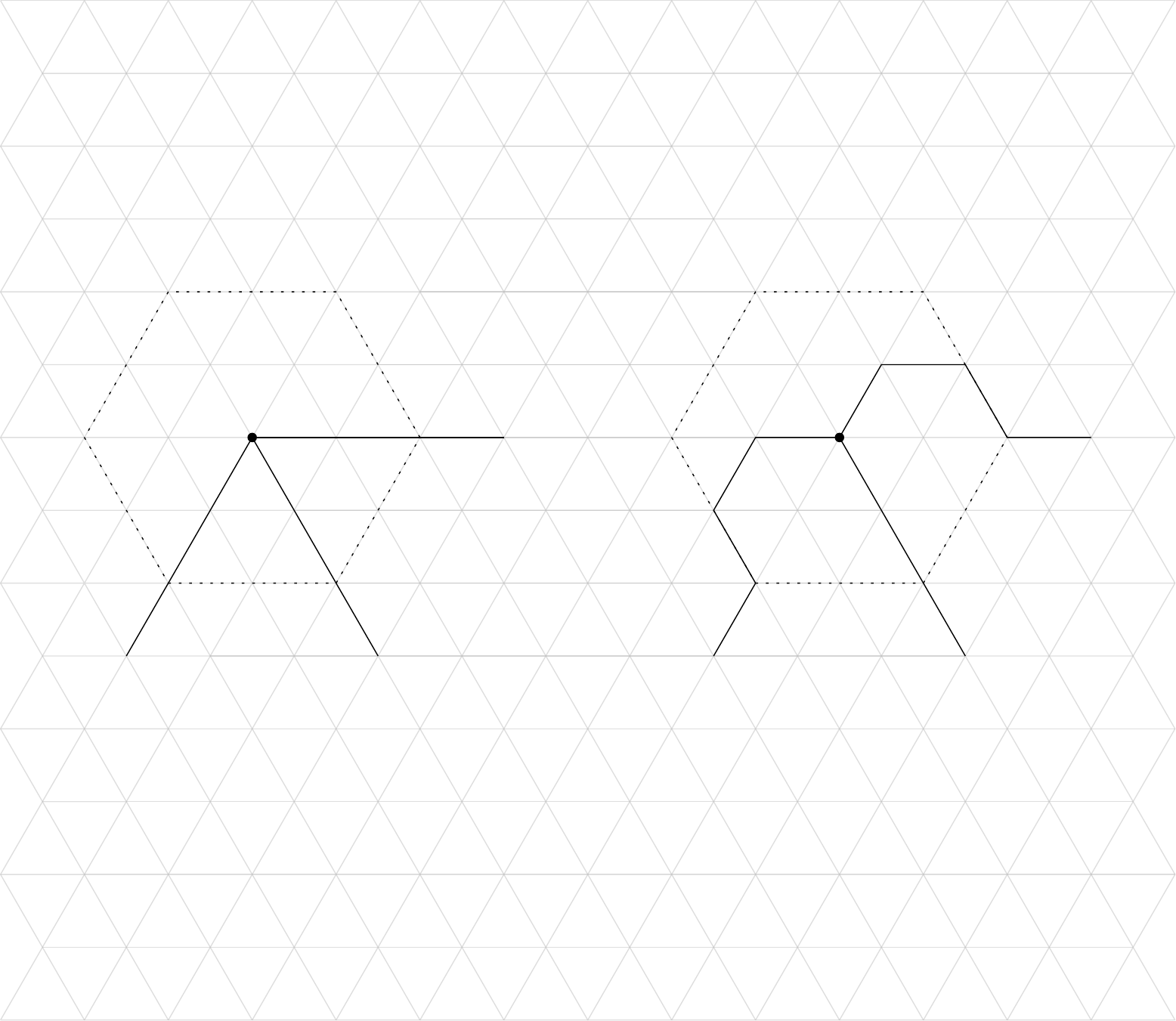}
\end{center}
\caption{Increasing angular resolution in the neighborhood of a vertex}
\label{fig:resol_inc}
\end{figure}

Next, we introduce another scaling by a factor of six so that each grid edge
becomes at least six consecutive edges in the same direction. This way we get
a graph which is still a spanning subgraph of the triangular grid (due to the
angular resolution being $\frac{2\pi}{3}$), and the path between any pair of
original vertices is represented by a path whose length is a multiple of six.

\begin{figure}
\begin{center}
\includegraphics[trim={1.5cm 4.5cm 6cm
0.3cm},clip,scale=0.6]{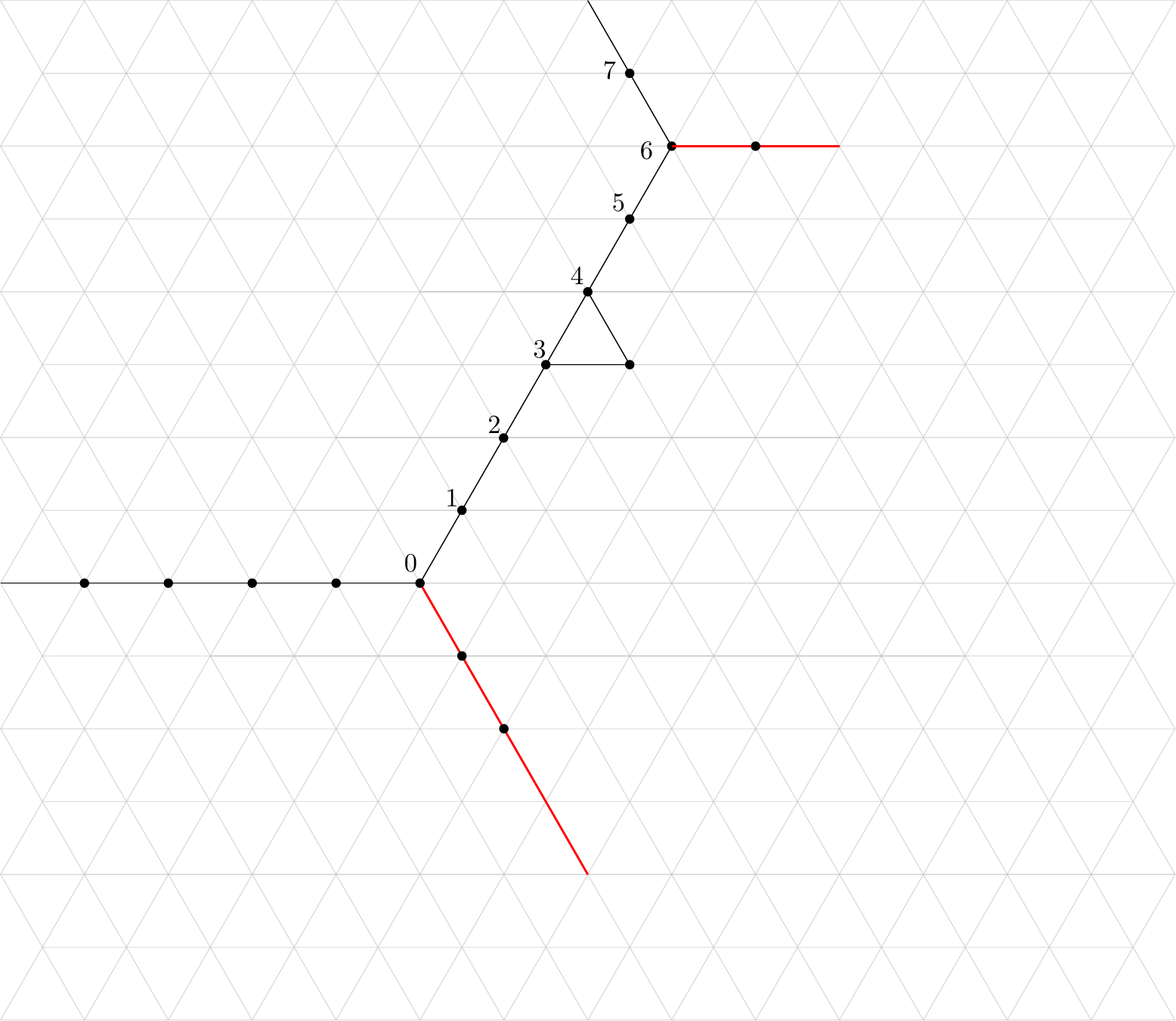}
\end{center}
\caption{Adding an ear to a variable cycle. The black edges are edges of the
variable circle, the red edges are on the literal paths.}
\label{fig:cycle_modify}
\end{figure}

\begin{figure}
\begin{center}
\includegraphics[trim={1cm 7.2cm 1cm 4.3cm},clip,scale=0.7]{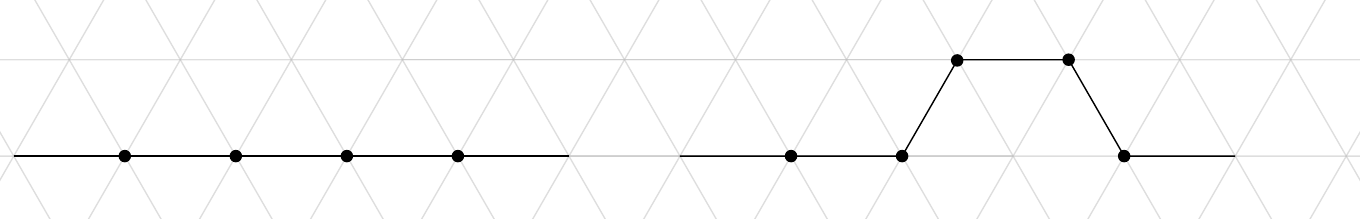}
\end{center}
\caption{Adding a detour to the path of a positive literal}
\label{fig:path_modify}
\end{figure}

\begin{figure}
\begin{center}
\includegraphics[trim={1.5cm 4.5cm 6cm 0.3cm},clip,scale=0.6]{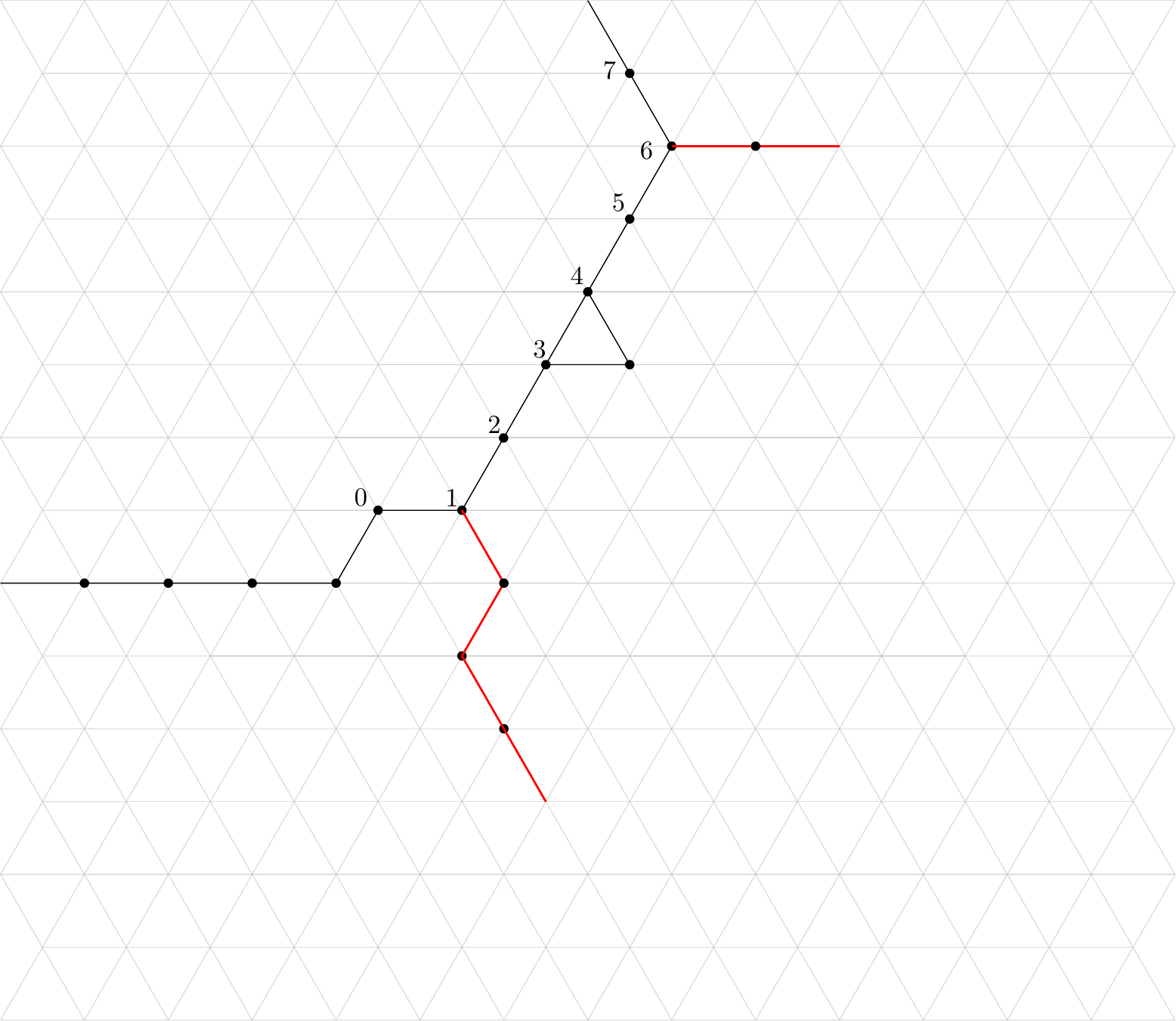}
\end{center}
\caption{Modifying the cycle connection of a negative literal. The black
edges are edges of the variable circle, the red edges are on the literal
paths.}
\label{fig:neg_modify}
\end{figure}

We make some extra modifications to the resulting graph. On each variable
cycle, we number the vertices consecutively in way that the paths starting
from the cycle get a number that is divisible by six (note that the two
scalings make this possible), for a cycle $C$ of length $6k$, these vertices
are denoted by $C_0,C_6, \dots, C_{6(k-1)}$. We add an \emph{ear} (a path of
length two) that connects $C_3$ and $C_4$; see Figure~\ref{fig:cycle_modify}.
It is easy to check that this can be done by adding a small triangle, and
preserving the spanning triangular grid graph property.

We introduce a small detour of length $1$ on a literal path if the
corresponding literal is positive. The detour is made as depicted in
Figure~\ref{fig:path_modify} to preserve the induced property. Finally, if a
path corresponds to a negative literal, we introduce another modification at
its connection with the variable circle, which makes the path one longer, and
connects to the circle at a point whose number is congruent to one modulo
three. See Figure~\ref{fig:neg_modify} for an example. Notice that the cycle
length remains unchanged, and the induced property is preserved. We also
number each literal path $P$ of length $3k+1$ consecutively, starting from the
vertex $P_0$ which is shared with the variable cycle, to the last vertex
$P_{3k+1}$, which is the vertex representing the clause.

The resulting (induced) triangle grid graph is denoted by $H$. Notice that
$H$ can be constructed in polynomial time from a given 3-CNF formula.

Given a satisfying assignment, we can define a dominating set of size
$\frac{|E(H)|-3m}{3}$. This can be done by selecting vertices of the form
$C_{3\ell}$ on cycles of true variables and $C_{3\ell+1}$ on cycles of false
variables. On the path of true literals, vertices of the form $P_{3\ell}$ are
selected, and on the paths of false literals, vertices of the form $P_{3\ell+2}$
are selected. It is routine to check that this is a dominating set; we note
that the clause vertices are dominated by the last inner point of the path of
their true literal(s).

\begin{lemma}\label{lem:ds_lowerbound}
A dominating set of $H$ has at least $\frac{|E(H)|}{3}-m$ points; more
specifically, a dominating set covers at least $k$ points of
\begin{itemize}
\item the vertices of a variable cycle of length $3k$;
\item the inner vertices of a variable path of length $3k+1$.
\end{itemize}
\end{lemma}

\begin{proof}
We consider a cycle first. Notice that vertices of the form $3\ell+2$ have
only neighbors inside the cycle; the neighborhoods of these vertices are
disjoint, therefore each of these neighborhoods along with the vertex itself
must contain a distinct dominating point. Since the number of such vertices
is $k$, every dominating set contains at least $k$ dominating points from the
cycle.

A similar argument applies on variable paths. The points of the form $3\ell+2
\quad (\ell=0 \dots k-1)$ only have neighbors that are inner vertices of the
path, and the neighborhoods are disjoint, which means that there must be at
least one dominating point in each of these $k$ neighborhoods.
\end{proof}

\begin{proof}[Proof of Theorem~\ref{thm:DSNPH}]
We have seen that $H$ can be constructed in polynomial time, it is sufficient
to prove that $H$ has a dominating set of size $\frac{|E(H)|}{3}-m$ if and
only if the original 3-CNF formula is satisfiable. We have already
demonstrated that there is a dominating set of this size if the formula is
satisfiable.

Let $D$ be a dominating set of size $\frac{|E(H)|}{3}-m$ in $H$. By
Lemma~\ref{lem:ds_lowerbound}, the number of points of $D$ in a literal path
of length $3k+1$ is exactly $k$, one point from each neighborhood of the
points $3\ell+2\quad (\ell=0 \dots k-1)$. Notice that the clause vertex can
not be in $D$. Suppose that vertex $P_1$ is in $D$. It follows that in the
neighborhood of vertex $P_5$, $P_4$ must be the one in $D$, otherwise $P_3$
would be uncovered. Continuing this pattern along the path we see that the
vertices $P_{3\ell+1}$ from each neighborhood will be in $D$ --- but that 
leaves $P_{3k}$ uncovered. Thus, $P_1$ cannot be in $D$.

Now take a variable cycle of length $3k$. The number of dominating points on
the cycle is exactly $k$, one point in each neighborhood around the points
$3\ell+2$. Notice that the ear means that one of these points is congruent to
$0$ or $1$ modulo three (since the ear vertex has to be dominated by $C_3$ or
$C_4$). We also know that $P_1$ on a connecting literal path $Q$ is not
in $D$, so we can use a similar strategy as on the paths to verify that all
the points of $D$ in $C$ are in fact congruent modulo three. We assign the
variable TRUE if these points are congruent to $0$, and FALSE if they are
congruent to $1$.

Now looking at any clause vertex, it must be dominated by the last inner
vertex $P_{3k}$ of at least one literal path; going back along the path it
follows the points of the form $P_{3\ell}$ are in $D$, thus $P_0 \in D$.
Since this point is on the variable cycle, it means that in case of a
positive literal the variable is true, and in case of a negative literal the
variable is false. Therefore the formula is satisfied by our assignment.
\end{proof}

\section{A generalization of a proof by Marx}~\label{app:Marx_generic}

In this section, we show that the dominating set problem is \Wone-hard for intersection
graphs of square-like objects. We give an overview of the proof by
Marx~\cite{Marx06}, highlighting the main differences for square-like
objects.

Marx uses a reduction from \textsc{Grid Tiling}
\cite{ParamAlg15} (although he does not explicitly state it this way). In a
grid-tiling problem we are given an integer $k$, an integer $n$, and a
collection $\cS$ of $k^2$ non-empty sets $U_{a,b} \subseteq \civ{n} \times
\civ{n}$ for $1 \leq a,b \leq k$. The goal is to select an element
$u_{a,b}\in U_{a,b}$ for each $1 \leq a,b \leq k$ such that
\begin{itemize}
\setlength\itemsep{0em}
\item If $u_{a,b}=(x,y)$ and $u_{a+1,b}=(x',y')$, then $x=x'$.
\item If $u_{a,b}=(x,y)$ and $u_{a,b+1}=(x',x')$, then $y=y'$.
\end{itemize}
One can picture these sets in a $k\times k$ matrix: in each cell $(a,b)$, we
need to select a representative from the set $U_{a,b}$ so that the
representatives selected from horizontally neighboring cells agree in the
first coordinate, and representatives from vertically neighboring sets agree
in the second coordinate.

Let $Q$ be a square-like object, and $\bb_1,\bb_2,\bu_1,\bu_2$ the
corresponding vectors. The original reduction places $k^2$ gadgets, one for
each $U_{a,b}$. A gadget contains 16 blocks of $Q$-translates, labeled
$X_1,Y_1,X_2,Y_2,\dots,X_8,Y_8$, that are arranged along the edges of a
square---see Fig.~\ref{fig:W1-hardness}. For square-like objects, we place
the reference points of the blocks in the grid $k\bb_1+\ell\bb_2, \; k,\ell
\in \Ints$. Initially, each block $X_{\ell}$  contains $n^2$ $Q$-translates,
denoted by $X_\ell(1), \dots, X_\ell(n^2)$ and each block $Y_{\ell}$ contains
$n^2+1$ $Q$-translates denoted by $Y_\ell(0), \dots, Y_\ell(n^2)$. The
argument $j$ of $X_\ell(j)$ can be thought of as a pair $(x,y)$ with $1\le
x,y\leq n$ for which $f(x,y):=(x-1)n+y=j$. Let $f^{-1}(j) = \left(\iota_1(j),
\iota_2(j)\right) = \left(1+\lfloor j/n\rfloor,1+(j \mod n)\right)$.

%----------------------------------------------------------------------------
\begin{figure}[bt]
    \begin{center}
    \includegraphics[scale=0.8]{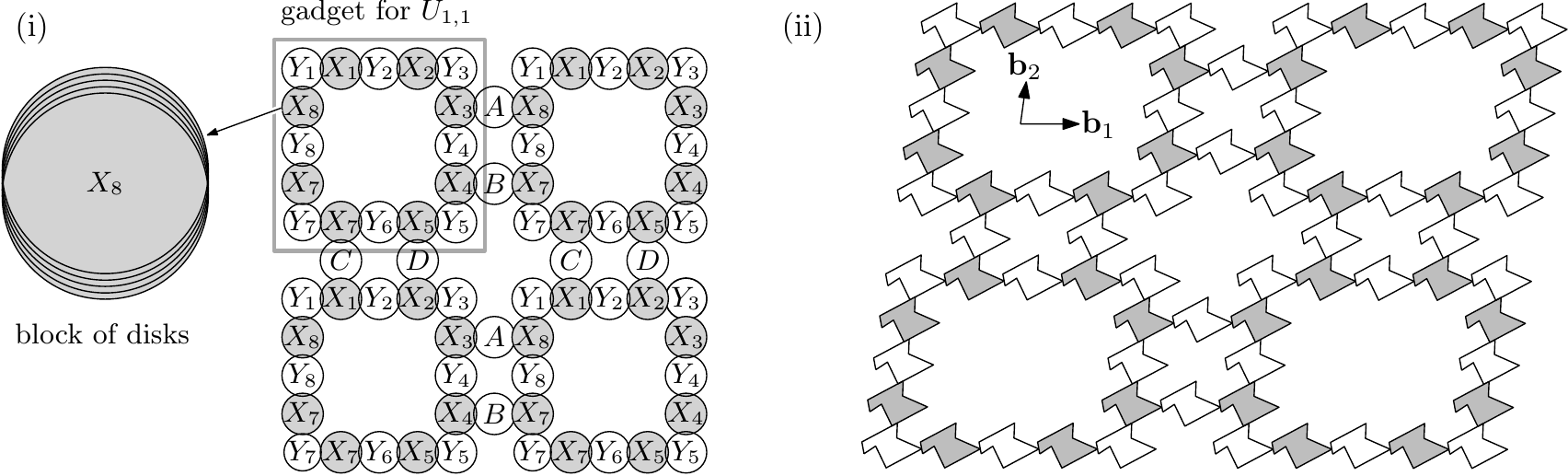}
    \end{center}
    \caption{(i) The construction by Marx for disks. (ii) Construction with translates of square-like shapes.}
    \label{fig:W1-hardness}
\end{figure}
%----------------------------------------------------------------------------
For the final construction, in each gadget at position $(a,b)$, delete all
$Q$-translates $X_\ell (j)$ for each $\ell=1,\dots,8$ and
$\left(\iota_1(j),\iota_2(j)\right) \not\in U_{a,b}$. This deletion ensures
that the gadgets represent the corresponding set $U_{a,b}$. The construction
is such that a minimum dominating set uses only $Q$-translates in the
$X$-blocks, and that for each gadget $(a,b)$ the same $Q$-translate
$X_{\ell}(j)$ is chosen for each $1\leq \ell\leq 8$. This choice signifies a
specific choice $u_{a,b}=(x,y)$. To ensure that the choice for $u_{a,b}$ in
the same row and column agree on their first and second coordinate,
respectively, there are special connector blocks between neighboring gadgets.
The connector blocks are denoted by $A,B,C$ and $D$ in
Fig.~\ref{fig:W1-hardness}, and they each contain $n+1$ $Q$-translates.

\mypara{Defining the blocks.} In every block, the place of each $Q$-translate
is defined with regard to the reference point of the block, $z$. The
reference point of each $Q$-translate is of the form
$z+\alpha\bu_1+\beta\bu_2$ where $\alpha$ and $\beta$ are integers. We say
that the \emph{offset} of this $Q$-translate is $(\alpha,\beta)$. The offsets
of $X$ and $Y$-blocks are defined as follows.

\begin{center}
\begin{tabular}{ll}
offset$(X_1(j))=( j,-\iota_2(j))$ & offset$(Y_1(j))=(j+0.5,j+0.5)$ \\

offset$(X_2(j))=( j, \iota_2(j))$ & offset$(Y_2(j))=(j+0.5,-n)$ \\

offset$(X_3(j))=(-\iota_1(j),-j)$ & offset$(Y_3(j))=(j+0.5,-j-0.5)$ \\

offset$(X_4(j))=( \iota_1(j),-j)$ & offset$(Y_4(j))=(-n,-j-0.5)$ \\

offset$(X_5(j))=(-j, \iota_2(j))$ & offset$(Y_5(j))=(-j-0.5,-j-0.5)$ \\

offset$(X_6(j))=(-j,-\iota_2(j))$ & offset$(Y_6(j))=(-j-0.5,n)$ \\

offset$(X_7(j))=( \iota_1(j), j)$ & offset$(Y_7(j))=(-j-0.5,j+0.5)$ \\

offset$(X_8(j))=(-\iota_1(j), j)$ & offset$(Y_8(j))=(n,j+0.5)$ \\
\end{tabular}
\end{center}

We remark some important properties. First, two $Q$-translates can intersect
only if they are in the same or in neighboring blocks. Consequently, one
needs at least eight $Q$-translates to dominate a gadget. The second
important property is that $Q$-translate $X_\ell(j)$ dominates exactly
$Y_{\ell}(j),\dots,Y_{\ell}(n^2)$ from the ``previous'' block $Y_\ell$, and
$Y_{\ell+1}(0),\dots,Y_{\ell+1}(j-1)$ from the ``next'' block $Y_{\ell+1}$.
This property can be used to prove the following key lemma.

\begin{lemma}[Lemma 1 of~\cite{Marx06}]\label{lem:Mgadget}
Assume that a gadget is part of an instance such that none of the blocks
$Y_i$ are intersected by $Q$-translates outside the gadget. If there is a
dominating set $\Delta$ of the instance that contains exactly $8k^2$
$Q$-translates, then there is a \emph{canonical} dominating set $\Delta'$
with $|\Delta'| = |\Delta|$, such that for each gadget $\cG$, there is an
integer $1 \leq j^G \leq n$ such that $\Delta'$ contains exactly the
$Q$-translates $X_1(j^G), \dots, X_8(j^G)$ from $\cG$.
\end{lemma}

In the gadget $G_{a,b}$, the value $j$ defined in the above lemma represents
the choice of $s_{a,b}=\left(\iota_1(j),\iota_2(j)\right)$ in the grid tiling
problem. Our deletion of certain $Q$-translates in $X$-blocks ensures that
$\left(\iota_1(j),\iota_2(j)\right) \in U_{a,b}$. Finally, in order to get a
feasible grid tiling, gadgets in the same row must agree on the first
coordinate, and gadgets in the same column must agree on the second
coordinate. These blocks have $n+1$ $Q$-translates each, with indices
$0,1,\dots,n$. We define the offsets in the connector gadgets the following
way.

\begin{center}
\begin{tabular}{ll}
offset$(A_j)=(-j-0.5,-n^2-1)$ & offset$(B_j)=(j+0.5,n^2+1)$ \\

offset$(C_j)=(n^2+1,-\iota_2(j))$ & offset$(D_j)=(-n^2-1,\iota_2(j))$ \\
\end{tabular}
\end{center}

Using this definition, it is easy to prove the following lemma.

\begin{lemma}\label{lem:Mconnect}
Let $\Delta$ be a canonical dominating set. For ``horizontally'' neighboring
gadgets $G$ and $H$ representing $j_G$ and $j_H$, the $Q$-translates of the
connector block $A$ are dominated if and only if $\iota_1(j_G)\le
\iota_1(j_H)$; the $Q$-translates of $B$ are dominated if and only if
$\iota_1(j_G)\ge \iota_1(j_H)$. Similarly, for vertically neighboring blocks
$G'$ and $H'$, the $Q$-translates of block $C$ are dominated if and only if
$\iota_2(j_{G'})\leq \iota_2(j_{H'})$; the $Q$-translates of $D$ are
dominated if and only if $\iota_2(j_{G'})\ge \iota_2(j_{H'})$.
\end{lemma}

With the above lemmas, the correctness of the reduction follows. A feasible
grid tiling defines a dominating set of size $8k^2$: in gadget $G_{a,b}$, the
dominating $Q$-translates are $X_\ell\left(f(s_{a,b})\right),
\;\ell=1,\dots,8$. On the other hand, if there is a dominating set of size
$8k^2$, then there is a canonical dominating set of the same size  that
defines a feasible grid tiling.
\end{document}